\documentclass[a4paper,onecolumn,11pt,accepted=2020-08-24]{quantumarticle}
\pdfoutput=1
\usepackage{amsmath,amsfonts,amssymb,amsthm}
\usepackage{graphicx, amscd}
\usepackage{color,comment}
\usepackage[T1]{fontenc}
\usepackage[latin1]{inputenc}
\usepackage[toc,page]{appendix}
\usepackage[colorlinks=true,allcolors=blue,]{hyperref}

\usepackage{tikz}
\usetikzlibrary{decorations.pathreplacing,angles,quotes}

\usepackage[
backend=bibtex,%
style=trad-alpha,%
doi=true,%
url=true,%
]{biblatex}
\bibliography{Weak_designs_revised.bib}

\newtheorem{theorem}{Theorem}[section]
\newtheorem{lemma}[theorem]{Lemma}
\newtheorem{corollary}[theorem]{Corollary}
\newtheorem{proposition}[theorem]{Proposition}
\newtheorem{definition}[theorem]{Definition}
\newtheorem{remark}[theorem]{Remark}
\newtheorem{defi/prop}[theorem]{Definition/Proposition}

\newcommand{\N}{\mathbf{N}}

\newcommand{\R}{\mathbf{R}}
\newcommand{\C}{\mathbf{C}}
\renewcommand{\P}{\mathbf{P}}

\renewcommand{\leq}{\leqslant}
\renewcommand{\geq}{\geqslant}

\DeclareMathOperator{\Tr}{Tr}
\DeclareMathOperator{\E}{\mathbf{E}}

\newcommand{\braket}[2]{\langle #1 | #2\rangle}
\newcommand{\ketbra}[2]{| #1 \rangle\!\langle #2 |}
\newcommand{\bra}[1]{\langle #1 |}
\newcommand{\ket}[1]{| #1 \rangle}

\newcommand{\Enc}{\mathrm{Enc}}
\newcommand{\Dec}{\mathrm{Dec}}

\definecolor{dgreen}{rgb}{.05,.45,.05}

\title{Weak approximate unitary designs and applications to quantum encryption}

\date{August 25 2020}

\author{C\'{e}cilia Lancien}
\affiliation{Institut de Math\'{e}matiques de Toulouse \& CNRS, Universit\'{e} Paul Sabatier, 118 route de Narbonne, F-31062 Toulouse Cedex 9, France.}
\email{clancien@math.univ-toulouse.fr}
\author{Christian Majenz}
\affiliation{QuSoft and Centrum Wiskunde \& Informatica, Science Park 123, 1098 XG Amsterdam, the Netherlands.}
\email{christian.majenz@cwi.nl}

\begin{document}

\maketitle

\begin{abstract}
	Unitary $t$-designs are the bread and butter of quantum information theory and beyond. An important issue in practice is that of efficiently constructing good approximations of such unitary $t$-designs. Building on results by Aubrun (Comm.~Math.~Phys.~2009), we prove that sampling $d^t\mathrm{poly}(t,\log d, 1/\epsilon)$ unitaries from an exact $t$-design provides with positive probability an $\epsilon$-approximate $t$-design, if the error is measured in one-to-one norm. As an application, we give a randomized construction of a quantum encryption scheme that has roughly the same key size and security as the quantum one-time pad, but possesses the additional property of being non-malleable against adversaries without quantum side information.
\end{abstract}

\section{Introduction} 

 Random unitaries, drawn from the Haar measure on the unitary group, play an important role in many aspects of theoretical quantum information science. For instance, most results on quantum source and channel coding are obtained with Haar-random coding strategies \cite{abeyesinghe2009mother,Hayden2008,Berta2011} using the decoupling technique \cite{Horodecki2007,Szehr2013,Dupuis2014,Majenz2017}. The columns and rows of Haar random unitaries are Haar random unit vectors and have also found many applications in quantum information theory, e.g.~for constructing quantum money schemes \cite{Ji2019,Alagic2019}. However Haar random unitaries are infeasible to even approximate, the randomness and number of gates necessary to sample and implement them being exponential in the number of qubits they act on. 
 
 In most situations, unitary $t$-designs, the quantum analogues of $t$-wise independent functions, come to the rescue \cite{Dankert2009}. A unitary $t$-design is a measure on the unitary group that reproduces the Haar measure up to the $t$-th moment. This means that a random unitary sampled from a $t$-design can replace a Haar-random unitary in any situation where it is only applied $t$ times. For practical purposes, one would like this measure to be more economical than the Haar measure (for instance to have finite, as small as possible, support). Often even just approximate versions of unitary $t$-designs (in the right metrics) are sufficient. In quantum information theory and related fields the most common metric between measures on the unitary group is the completely bounded one-to-one norm, or diamond norm, on the induced $t$-twirling channels. The $t$-twirling channel associated to a measure is the channel that can be implemented by sampling a unitary according to the measure, and then applying it to each sub-system of a $t$-partite input system.
 
 In \cite{Hayden2004}, approximate $1$-designs have been studied using a metric based on the (not completely bounded) one-to-one norm. There, it is shown that approximate $1$-designs in this weaker sense can be made of much less unitaries, and that they still have interesting applications, such as unconditionally secure encryption of quantum data when confidentiality is only desired against adversaries without quantum side information. The former result is shown by proving that sampling a small number of independent Haar-random unitaries provides with high probability an approximate $1$-design. This construction was subsequently partially derandomized in \cite{Aubrun2009}, where it was shown that sampling from a measure which is only a $1$-design works as well.
 
 Let us mention one last result which was known prior to this work. It was shown in \cite{Lancien2017} that, in fact, any channel can be approximated in one-to-one norm by a channel having few Kraus operators. However, this does not tell us whether it can be further imposed that the Kraus operators of this approximating channel are of a specific form (such as e.g.~being unitaries sampled from a simple enough distribution, which is what we are interested in here).

 \subsubsection*{Our contribution}
 
 In this work, we generalize the approach of \cite{Aubrun2009} to construct small approximate $t$-designs, for any given $t$, in one-to-one norm distance. In addition, for $t=2$, we show that the approach extends to designs where the goal is to approximate the channel twirl, i.e.~the transformation of quantum channels obtained by sampling a unitary, applying it to the input state before the channel acts on it, and undoing this action afterwards. Here, the appropriate distance is the one stemming from the operator norm induced by the diamond norm, which we call diamond-to-diamond norm. To prove the approximation result on the so-called $U^{\otimes t}$-twirl, we use basic representation theory of the unitary group, including the Weyl dimension formula, to show that it has small one-to-operator norm. This allows us to apply the powerful probabilistic and functional analytic tools developed in \cite{Aubrun2009}. For the channel twirl, the invariant space spanned by the identity, as well as the off-diagonal terms involving this invariant space, require a careful analysis. Along the way, we also construct a design that approximates the so-called $U\otimes\bar U$-twirl, the image of the channel twirl under the Choi-Jamio\l kowski isomorphism.
 
 What is more, we prove that our results are optimal, in the following sense: approximating the twirling channels under consideration cannot be done with less operators than what our sub-sampling approach gives, even without imposing any structure on them (in our case the constraint of being a tensor product of unitaries). 
 
 \subsubsection*{An application} 
 
 Subsequently, we apply our results in a cryptographic context. We show, that an approximate channel-twirl design in the diamond-to-diamond norm metric can be used to construct a quantum encryption scheme that is as secure as the quantum one-time pad and has (essentially) the same key length, but also is non-malleable against adversaries without quantum side information. While the construction is not time-efficient, it provides theoretical insights, and constitutes evidence that savings in key size are possible. In particular, the construction quantifies in a precise way the amount of secret key that a full two-design-based non-malleable quantum encryption scheme uses just to counter side information attacks.
 
 Beyond applications to cryptography, the Kraus rank of a quantum channel can be considered, more generally, as a measure of its complexity. It indeed quantifies the minimal amount of ancillary resources needed to implement it. Equivalently, it quantifies the amount of degrees of freedom in it that one is ignorant of. It is thus natural to ask: given a quantum channel, is it possible to reduce its complexity while not affecting  its action too much? Or, in other words, is it possible to find a channel with much smaller Kraus rank which approximately simulates it? In our case, we further impose that the Kraus operators of the approximating channel, in addition to being few, inherit the structure of those of the original channel. Our results can therefore be seen as statements about complexity reduction of twirling channels, under extra constraints. As explained in \cite{Lancien2017}, results of this type provide, amongst other, efficient schemes for the destruction of correlations and data hiding in bipartite states.
 
\subsubsection*{Related work}

Unitary $t$-designs exist for all $t$ and all dimensions \cite{Seymour1984,Kane2015}\footnote{In \cite{Kane2015}, the existence of exact designs is proven in a much more general context, see \cite[Corollary 2]{Alagic2019} for a straightforward application to the unitary case.}. For $t>3$, time-efficient constructions are, however, only known for approximate unitary $t$-designs \cite{Brandao2016}. An appealing approach to try and exhibit unitary $t$-designs would be to look for them amongst unitary groups, equipped with their uniform measure. For $t\leq 3$ the Clifford group is known to be such a unitary $t$-group \cite{Webb2016,zhu2017}. Nevertheless, it was recently proved in \cite{bannai2020} that there is no unitary $t$-group for $t\geq 4$ (except in dimension $2$), so that this strategy cannot work anymore. 
The sub-sampling technique that we use, following \cite{Aubrun2009}, i.e.~the strategy of sampling a random subset of unitaries from an exact design, was first introduced in \cite{Ambainis2009} to show the existence of small approximate $2$-designs.

Non-malleability for quantum encryption was first introduced and characterized in \cite{Ambainis2009}. In this work it was also shown that the notion of quantum  non-malleability is equivalent to the notion of approximate unitary $2$-designs, under the condition that the encryption algorithm be unitary. Subsequently, non-malleability for quantum encryption has been further studied in \cite{Alagic2017,Majenz2019}.

\subsubsection*{Notation and standard definitions} 

Let us gather here notation that we will be using throughout the whole paper. Given $d\in\N$, we denote by $L(d)$ the set of linear operators on $\C^d$, by $D(d)$ the set of quantum states (i.e.~positive semidefinite and trace $1$ operators) on $\C^d$, and by $U(d)$ the set of unitary operators on $\C^d$. We additionally denote by $\mathcal{L}(d)$ the set of linear operators on $L(d)$, and by $\mathcal{C}(d)$ the set of quantum channels (i.e.~completely positive and trace-preserving operators) on $L(d)$. Let us conclude with a some standard notation/definitions from probability theory. Given a random variable $X$, we denote by $\E X$ its average and by $\P(X\in E)$ the probability that $X$ satisfies event $E$. We say that $\varepsilon$ is a Bernoulli random variable if $\P(\varepsilon=+1)=\P(\varepsilon=-1)=1/2$.

\section{Representation theoretic preliminaries} 

A good introduction to the representation-theoretic concepts used in this work can be found in \cite{fulton2013representation} (see also \cite{Christandl} for a short introduction that is very accessible for quantum information theorists). Given $t\in\N$ let $S_t$ be the permutation group of $\{1,\ldots,t\}$. The irreducible representations $[\lambda]$ of $S_t$ are called \emph{Specht modules} and are indexed by integer partitions of $t$, denoted as $\lambda\vdash t$. Such a partition is represented as a tuple $\lambda=(\lambda_1,...,\lambda_r)\in\N^r$, for some $r\in\N$, with $\lambda_1\geq \cdots\geq\lambda_r$ and $\sum_{i=1}^r\lambda_i=t$.

Given $d\in\N$ let $U(d)$ be the unitary group of $\C^d$. The polynomial irreducible representations $V_\lambda$ of $U(d)$ are called \emph{Weyl modules} and are indexed by integer partitions of any number $t\in\N$ into exactly $d$ parts (some of which might be $0$), denoted as $\lambda\vdash(t,d)$. The dimension of the Weyl module $V_\lambda$ is given by the Weyl dimension formula
\begin{equation}\label{eq:weyl}
	m_\lambda=\prod_{\substack{i,j\in\{1,...,d\}\\i<j}}\frac{\lambda_i-\lambda_j+j-i}{j-i}.
\end{equation}

A particular vector space that carries representations of both $S_t$ and $U(d)$ is $(\C^d)^{\otimes t}$, the corresponding actions are defined as
\begin{align*}
	& \forall\ \sigma\in S_t,\ \sigma.\ket{\phi_1}\otimes\cdots\otimes\ket{\phi_t} = \ket{\phi_{\sigma^{-1}(1)}}\otimes\cdots\otimes\ket{\phi_{\sigma^{-1}(t)}}, \\
	& \forall\ U\in U(d),\ U.\ket\phi = U^{\otimes t}\ket\phi.
\end{align*}
The two actions commute, i.e.~$(\C^d)^{\otimes t}$ decomposes into a direct sum of irreducible representations (irreps) of the product group $S_t\times U(d)$. These irreps are just tensor products of an irrep of $S_t$ with an irrep of $U(d)$. What is more, the corresponding representations of the group algebras of $S_t$ and $U(d)$ are double commutants, implying that the decomposition is multiplicity free.
\begin{theorem}[Schur-Weyl duality]
	Let $S_t$ and $U(d)$ act on $(\C^d)^{\otimes t}$ as described above. The direct sum decomposition into irreducible representations of $S_t\times U(d)$ is multiplicity free, and is given by
	\begin{equation}\label{eq:schurweyl}
	(\C^d)^{\otimes t}\cong\bigoplus_{\lambda\vdash(t,d)}V_\lambda\otimes[\lambda].
	\end{equation}
\end{theorem}

Define the quantum channel $T^{(t)}$ on $(\C^d)^{\otimes t}$ as 
\begin{equation} \label{eq:T^t}
T^{(t)}: X\in L(d^t) \mapsto \int_{U\in U(d)} U^{\otimes t}XU^{*\otimes t} dU \in L(d^t),
\end{equation}
where $dU$ stands for the Haar measure on $U(d)$. The channel $T^{(t)}$ is often referred to as a \textit{twirling channel}. It is obviously covariant with respect to the action of $U(d)$. Hence, denoting by $W$ the isomorphism between the right and left hand sides of the equation \eqref{eq:schurweyl} above, Schur's Lemma implies that
\begin{equation}\label{eq:twirl-char}
	W T^{(t)}(W^*(\cdot)W)W^*=\sum_{\lambda\vdash(t,d)}\tau_{V_\lambda}\otimes\Tr_{V_\lambda}\left[P_\lambda(\cdot)P_\lambda\right],
\end{equation}
where $P_\lambda$ is the projector onto $V_\lambda\otimes[\lambda]$ in $\bigoplus_{\lambda\vdash(t,d)}V_\lambda\otimes[\lambda]$ and $\tau_{V_\lambda}=\mathbf 1_{V_\lambda}/m_\lambda$ is the maximally mixed state on $V_\lambda$.

Let us make things slightly more explicit in the case $t=2$. We have
\[ (\C^d)^{\otimes 2}\cong \wedge^2(d) \oplus \vee^2(d), \]
where $\wedge^2(d)$ and $\vee^2(d)$ are, respectively, the symmetric and anti-symmetric subspaces of $(\C^d)^{\otimes 2}$. The corresponding projectors are $P_{\wedge^2(d)}=(\mathbf{1}+F)/2$ and $P_{\vee^2(d)}=(\mathbf{1}-F)/2$, where $F$ denotes the so-called \textit{flip operator}. And the action of $T^{(2)}$ can be explicitly written as, for all $X\in L(d^2)$,
\[ T^{(2)}(X) = \frac{2}{d(d+1)}\Tr\left(P_{\wedge^2(d)}X P_{\wedge^2(d)}\right)P_{\wedge^2(d)} + \frac{2}{d(d-1)}\Tr\left(P_{\vee^2(d)}X P_{\vee^2(d)}\right)P_{\vee^2(d)}. \]

Fix a basis $B=\{\ket{i}\}_{i=0}^{d-1}$ for $\C^d$ and let $\mathrm{T}$ be the transposition in this basis. It is easy to check that, denoting by $X^{\Gamma}$ the partial transposition of $X$ (i.e.~$X^{\Gamma}=\mathrm{id}\otimes\mathrm{T}(X)$), we have, for all $X\in L(d^2)$,
\[ T^{(2)}(X)^{\Gamma} = \left( \int_{U\in U(d)} U\otimes U X U^*\otimes U^* dU \right)^{\Gamma} = \int_{U\in U(d)} U\otimes \bar{U} X^{\Gamma} U^*\otimes \bar{U}^* dU . \]
Let us define the quantum channel $T^{(1,1)}$ on $(\C^d)^{\otimes 2}$ as
\begin{equation} \label{eq:T^11}
T^{(1,1)}: X\in L(d^2) \mapsto \int_{U\in U(d)} U\otimes \bar{U} X U^*\otimes \bar{U}^* dU . 
\end{equation}
By the preceding discussion, we know that $T^{(1,1)}(X)$ can be written as a linear combination of $P_{\wedge^2(d)}^{\Gamma}$ and $P_{\vee^2(d)}^{\Gamma}$. Now, $\mathbf{1}^{\Gamma}=\mathbf{1}$ and $F^{\Gamma}=d\ketbra{\psi}{\psi}$, where 
\[\ket{\psi}=\frac 1{\sqrt{d}}\sum_{i=0}^{d-1}\ket i\]
 is the  \textit{standard maximally entangled state} with respect to $B$. So equivalently, $T^{(1,1)}(X)$ can be written as a linear combination of $\ketbra{\psi}{\psi}$ and $Q=\mathbf{1}-\ketbra{\psi}{\psi}$, which are orthogonal to one another. More specifically, for all $X\in L(d^2)$,
\begin{equation}  \label{eq:T^11-action}
T^{(1,1)}(X)= \bra{\psi}X\ket{\psi} \ketbra{\psi}{\psi} + \frac{1}{d^2-1}\Tr(QX)Q .
\end{equation}

\section{Several channel approximation results}

Before we present our various twirling channel approximation results, let us state here the key technical lemma which is the starting point of most of our proofs. This lemma first appeared as \cite[Lemma 5]{Aubrun2009}. Its proof consists in estimating the average of the supremum of an empirical process through covering numbers, thanks to Dudley's inequality and a duality argument for entropy numbers. 

\begin{lemma}[{\cite[Lemma 5]{Aubrun2009}}] \label{lem:aubrun}
Let $U_1,\ldots,U_n\in U(d)$ and let $\varepsilon_1,\ldots,\varepsilon_n$ be independent Bernoulli random variables. Then,
\[ \E \left( \underset{\rho\in D(d)}{\sup} \left\| \sum_{i=1}^n \varepsilon_i U_i\rho U_i \right\|_{\infty} \right) \leq C(\log d)^{5/2}(\log n)^{1/2} \underset{\rho\in D(d)}{\sup} \left\| \sum_{i=1}^n U_i\rho U_i \right\|_{\infty}^{1/2}, \]
where $C>0$ is a universal constant.    
\end{lemma}

\subsection{Approximating the twirling channel $T^{(t)}$} 

Let $t\in\N$ be such that $t<d$. The goal here is to show that the twirling channel $T^{(t)}$, as defined by equation \eqref{eq:T^t}, can be approximated with `few' Kraus operators sampled from a `simple' probability measure. We will be able to prove such approximation in a strong sense, namely in one-to-infinity norm.

A probability measure $\mu$ on $U(d)$ is called a $t$-design if
\[ \forall\ X\in L(d^t),\ \int_{U\in U(d)} U^{\otimes t}XU^{*\otimes t} d\mu(U)= T^{(t)}(X). \]

We will show the following result:
\begin{theorem} \label{th:tdesign-T^t}
	Let $0<\epsilon<1$. Assume that the probability measure $\mu$ on $U(d)$ is a $t$-design, and let $U_1,\ldots,U_n$ be sampled independently from $\mu$. There exists a universal constant $C>0$ such that, if $n\geq C(td)^t(t\log d)^6/\epsilon^2$, then with probability at least $1/2$, we have
	\[ \forall\ \rho\in D(d^t),\ \left\| \frac{1}{n}\sum_{i=1}^n U_i^{\otimes t} \rho U_i^{*\otimes t} - T^{(t)}(\rho) \right\|_{\infty} \leq \frac{\epsilon}{d^t}. \]
\end{theorem}

Theorem \ref{th:tdesign-T^t} generalizes \cite[Theorem 2]{Aubrun2009} to $t$-designs for any $t\in\N$ rather than only for $1$-designs. We actually follow the exact same proof strategy as that of \cite[Theorem 2]{Aubrun2009}. The only additional technical lemma that we need in the case $t>1$ is one that tells us that $T^{(t)}$ has a small $(1\to\infty)$-norm (a fact which is obvious for $t=1$).
\begin{lemma} \label{lem:inftynorm-T^t}
The quantum channel $T^{(t)}$ is such that
\[ \underset{\rho\in D(d^t)}{\sup} \left\|T^{(t)}(\rho)\right\|_{\infty} \leq \left(\frac{2t}{d}\right)^t. \]
\end{lemma}
\begin{proof}
	By equation \eqref{eq:twirl-char}, the operator norm in question is just given by the inverse of the minimal dimension of an irrep $V_\lambda$,
	\begin{equation*}
		\underset{\rho\in D(d^t)}{\sup} \left\|T^{(t)}(\rho)\right\|_{\infty}=\frac{1}{\min_{\lambda\vdash(t,d)} m_\lambda}.
	\end{equation*}
	Indeed, let us denote by $\lambda^*$ the partition minimizing $m_{\lambda}$. It is clear that if $\ket{\phi_{\lambda^*}}\in V_{\lambda^*}$ and $\ket{\varphi_{\lambda^*}}\in [\lambda^*]$, then $\left\|T^{(t)}(\ketbra{\phi_{\lambda^*}}{\phi_{\lambda^*}}\otimes \ketbra{\varphi_{\lambda^*}}{\varphi_{\lambda^*}})\right\|_{\infty} = 1/m_{\lambda^*}$. And this is obviously maximizing $\left\|T^{(t)}(\rho)\right\|_{\infty}$ as $T^{(t)}$ begins with a pinching with respect to the direct sum decomposition \eqref{eq:schurweyl}.
	We go on to find a lower bound on $m_{\lambda^*}$ using the formula \eqref{eq:weyl}. To this end we first note that $\lambda^*$ is a partition of $t$ into $d$ parts, so $\lambda_i^*=0$ for all $i>t$. Noting that all the factors in the product in equation \eqref{eq:weyl} are lower bounded by $1$, and only keeping factors such that $i\le t<j$ we get
	\begin{align*} 
	m_{\lambda^*} & \ge \prod_{\substack{i,j\in\{1,...,d\}\\i\le t<j}}\frac{\lambda_i^*+j-i}{j-i} \\
	& = \prod_{i=1}^t\prod_{j=t+1}^d\frac{\lambda_i^*+j-i}{j-i} \\
	& = \prod_{i=1}^t\frac{(\lambda_i^*+d-i)!(t-i)!}{(\lambda_i^*+t-i)!(d-i)!} \\
	& = \prod_{i=1}^t\prod_{\alpha=1}^{\lambda_i^*}\frac{d-i+\alpha}{t-i+\alpha} .
	\end{align*}
	As a final step we use that $(a-x)/(b-x)\ge a/b$ for all $a,b,x\in \R$ such that $a\ge b>x\ge 0$, and that $\alpha\le t$. We thus conclude that 
	\[ \prod_{i=1}^t\prod_{\alpha=1}^{\lambda_i^*}\frac{d-i+\alpha}{t-i+\alpha} \ge \prod_{i=1}^t\prod_{\alpha=1}^{\lambda_i^*}\frac{d+\alpha}{t+\alpha}\nonumber \ge \prod_{i=1}^t\prod_{\alpha=1}^{\lambda_i^*}\frac{d}{2t} = \left(\frac{d}{2t}\right)^t. \]
\end{proof}

We then need the technical result below, which is an immediate corollary of \cite[Lemma 5]{Aubrun2009}, recalled earlier as Lemma \ref{lem:aubrun}. 
\begin{lemma} \label{lem:Bernoulli-T^t}
Let $U_1,\ldots,U_n\in U(d)$. For $\varepsilon_1,\ldots,\varepsilon_n$ independent Bernoulli random variables, we have
\[ \E \left( \underset{\rho\in D(d^t)}{\sup} \left\| \sum_{i=1}^n \varepsilon_i U_i^{\otimes t} \rho U_i^{*\otimes t} \right\|_{\infty} \right) \leq C(t\log d)^{5/2}(\log n)^{1/2} \underset{\rho\in D(d^t)}{\sup} \left\| \sum_{i=1}^n U_i^{\otimes t} \rho U_i^{*\otimes t} \right\|_{\infty}^{1/2}, \]
where $C>0$ is a universal constant.
\end{lemma}

\begin{proof} 
	This follows directly from \cite[Lemma 5]{Aubrun2009}, applied with $d^t$ playing the role of $d$ and $U_i^{\otimes t}$ playing the role of $U_i$, $1\leq i\leq n$.
\end{proof}

With these two preliminary lemmas at hand, we are now in position to prove Theorem \ref{th:tdesign-T^t}.
\begin{proof}[Proof of Theorem \ref{th:tdesign-T^t}]
Let $V_1,\ldots,V_n$ be independent copies of $U_1,\ldots,U_n$ and $\varepsilon_1,\ldots,\varepsilon_n$ be independent Bernoulli random variables. Setting
\[ M= \underset{\rho\in D(d^t)}{\sup} \left\| \frac{1}{n}\sum_{i=1}^n U_i^{\otimes t} \rho U_i^{*\otimes t} - T^{(t)}(\rho) \right\|_{\infty},  \]
we then have
\begin{align*}
\E M & = \E_U \left( \underset{\rho\in D(d^t)}{\sup} \left\| \frac{1}{n}\sum_{i=1}^n U_i^{\otimes t} \rho U_i^{*\otimes t} - \E_V\left( \frac{1}{n}\sum_{i=1}^n V_i^{\otimes t} \rho V_i^{*\otimes t} \right) \right\|_{\infty} \right) \\
& \leq \E_{U,V} \left( \underset{\rho\in D(d^t)}{\sup} \left\| \frac{1}{n}\sum_{i=1}^n \left( U_i^{\otimes t} \rho U_i^{*\otimes t} - V_i^{\otimes t} \rho V_i^{*\otimes t} \right) \right\|_{\infty} \right) \\
& = \E_{U,V,\varepsilon} \left( \underset{\rho\in D(d^t)}{\sup} \left\| \frac{1}{n}\sum_{i=1}^n \varepsilon_i \left( U_i^{\otimes t} \rho U_i^{*\otimes t} - V_i^{\otimes t} \rho V_i^{*\otimes t} \right) \right\|_{\infty} \right) \\
& \leq 2\, \E_{U,\varepsilon} \left( \underset{\rho\in D(d^t)}{\sup} \left\| \frac{1}{n}\sum_{i=1}^n \varepsilon_i U_i^{\otimes t} \rho U_i^{*\otimes t}\right\|_{\infty} \right),
\end{align*}
where the first inequality is by Jensen's inequality, the second equality is by symmetry, and the third inequality is by the triangle inequality.

Hence, by Lemma \ref{lem:Bernoulli-T^t}, we get
\begin{align*}
\E M & \leq \frac{2C}{\sqrt{n}}(t\log d)^{5/2}(\log n)^{1/2} \E \left( \underset{\rho\in D(d^t)}{\sup} \left\| \frac{1}{n}\sum_{i=1}^n U_i^{\otimes t} \rho U_i^{*\otimes t} \right\|_{\infty}^{1/2} \right) \\
& \leq \frac{2C}{\sqrt{n}}(t\log d)^{5/2}(\log n)^{1/2} \E \left( M+\left(\frac{2t}{d}\right)^t \right)^{1/2} \\
& \leq \frac{2C}{\sqrt{n}}(t\log d)^{5/2}(\log n)^{1/2} \left(\E \left(M+\left(\frac{2t}{d}\right)^t\right) \right)^{1/2}, \\
\end{align*}
where the second inequality is by Lemma \ref{lem:inftynorm-T^t} while the third inequality is by Jensen's inequality.

Now, it is easy to check that, given $X,\alpha,\beta\geq 0$, if $X\leq \alpha\sqrt{X+\beta}$, then $X\leq \alpha^2+\alpha\sqrt{\beta}$. Therefore, we eventually obtain
\[ \E M \leq \frac{4C^2}{n}(t\log d)^{5}\log n + \frac{2C}{\sqrt{n}}(t\log d)^{5/2}(\log n)^{1/2}\left(\frac{2t}{d}\right)^{t/2}. \]
And the latter quantity is smaller than $\epsilon/d^t$ as soon as $n$ is larger than $C'(td)^t(t\log d)^6/\epsilon^2$.

To conclude, we just have to use Markov's inequality, which guarantees that, if $\E M \leq\epsilon/d^t$, then
\[ \P\left(M\leq \frac{2\epsilon}{d^t} \right) \geq 1 - \frac{\E M}{2\epsilon/d^t} \geq \frac{1}{2}. \]
This is exactly what we wanted to show (after relabelling $2\epsilon$ in $\epsilon$ and $4C'$ in $C$).
\end{proof}

\begin{remark} \label{rem:optimality}
	Note that, up to a $\mathrm{poly}(t,\log d)$ factor, the result of Theorem \ref{th:tdesign-T^t} is optimal, in the sense that it is impossible to approximate the twirling channel $T^{(t)}$ with less than order $d^t$ operators. This is true even if we only require $\epsilon$-approximation in $(1\to 1)$-norm rather than $\epsilon/d^t$-approximation in $(1\to \infty)$-norm. The argument has a similar flavor as the one appearing in \cite[Section 5.1]{Lancien2017}, proving optimality of channel approximation in a more general setting.
	
	Indeed, let $T,\hat{T}$ be channels on $L(n)$ which are $\epsilon$-close in $(1\to 1)$-norm. Suppose that $T$ is such that, for all $\rho\in D(n)$, $\|T(\rho)\|_{\infty} \leq c/n$. Now, if $\hat{T}$ has Kraus rank $k<c/n$, then a pure input state $\rho$ is necessarily sent on an output state $\hat{T}(\rho)$ of rank at most $k$. Hence for $\rho\in D(n)$ pure, we have
	\[  \left\| T(\rho)-\hat{T}(\rho) \right\|_1 \geq \frac{n/c-k}{n/c} = 1-\frac{ck}{n} . \]
	But since by assumption we also have 
	\[  \left\| T(\rho)-\hat{T}(\rho) \right\|_1 \leq \epsilon , \]
	this means that necessarily $k\geq(1-\epsilon)n/c$
	
	In the case of the channel $T^{(t)}$ on $L(d^t)$, we know by Lemma \ref{lem:inftynorm-T^t} that, for all $\rho\in D(d^t)$, $\|T^{(t)}(\rho)\|_{\infty} \leq \left(2t/d\right)^t$. So if a channel is $\epsilon$-close to $T^{(t)}$ in $(1\to 1)$-norm, then it has to have Kraus rank at least $(1-\epsilon)(d/2t)^{t}$. 
\end{remark}

\subsection{Approximating the twirling channel $T^{(1,1)}$} 

The goal here is to show that the twirling channel $T^{(1,1)}$, as defined by equation \eqref{eq:T^11}, can be approximated with `few' Kraus operators sampled from a `simple' probability measure. We will only be able to prove such approximation in a weaker sense than in the case of $T^{(t)}$ treated before, namely in one-to-one norm.

If $\mu$ is a $2$-design on $U(d)$, then, by equation \eqref{eq:T^11}, we have that
\[ \forall\ X\in L(d^2),\ \int_{U\in U(d)} U\otimes\bar{U} X U^*\otimes\bar{U}^* d\mu(U)= T^{(1,1)}(X). \]

We will show the following result:
\begin{theorem} \label{th:11design-T^11}
	Let $0<\epsilon<1$. Assume that the probability measure $\mu$ on $U(d)$ is a $2$-design, and let $U_1,\ldots,U_n$ be sampled independently from $\mu$. There exists a universal constant $C>0$ such that, if $n\geq Cd^2(\log d)^6/\epsilon^2$, then with probability at least $1/2$, we have
	\[ \forall\ \rho\in D(d^2),\ \left\| \frac{1}{n}\sum_{i=1}^n U_i\otimes\bar{U}_i \rho U_i^*\otimes\bar{U}_i^* - T^{(1,1)}(\rho) \right\|_{1} \leq \epsilon. \]
\end{theorem}

The way we prove Theorem \ref{th:11design-T^11} is by first analysing separately the cases where the input state is the maximally entangled state or a state orthogonal to it. This is the content of Propositions \ref{prop:11design-T^11-1} and \ref{prop:11design-T^11-2} below.

\begin{proposition} \label{prop:11design-T^11-1}
Assume that the probability measure $\mu$ on $U(d)$ is a $2$-design, and let $U_1,\ldots,U_n$ be sampled independently from $\mu$. Then,
\[ \frac{1}{n}\sum_{i=1}^n U_i\otimes\bar{U}_i \ketbra{\psi}{\psi} U_i^*\otimes\bar{U}_i^* = T^{(1,1)}(\ketbra{\psi}{\psi}). \]
\end{proposition}

\begin{proof}
We just have to notice that, for any $U\in U(d)$, $U\otimes\bar{U} \ket{\psi} = \ket{\psi}$. And thus,
\[ \frac{1}{n}\sum_{i=1}^n U_i\otimes\bar{U}_i \ketbra{\psi}{\psi} U_i^*\otimes\bar{U}_i^* = \ketbra{\psi}{\psi} = T^{(1,1)}(\ketbra{\psi}{\psi}) , \]
as announced.
\end{proof}

\begin{proposition} \label{prop:11design-T^11-2}
	Let $0<\epsilon<1$. Assume that the probability measure $\mu$ on $U(d)$ is a $2$-design, and let $U_1,\ldots,U_n$ be sampled independently from $\mu$. There exists a universal constant $C>0$ such that, if $n\geq Cd^2(\log d)^6/\epsilon^2$, then with probability at least $1/2$, we have
	\[ \forall\ \rho\in D(d^2),\ \rho\perp\psi,\ \left\| \frac{1}{n}\sum_{i=1}^n U_i\otimes\bar{U}_i \rho U_i^*\otimes\bar{U}_i^* - T^{(1,1)}(\rho) \right\|_{\infty} \leq \frac{\epsilon}{d^2}. \]
\end{proposition}

In order to prove Proposition \ref{prop:11design-T^11-2} we follow the same route as to prove Theorem \ref{th:tdesign-T^t}. We thus begin by observing that $T^{(1,1)}$ has a small $(1\to\infty)$-norm on the orthogonal complement of the maximally entangled state, which is the analogue of Lemma \ref{lem:inftynorm-T^t} in the study of $T^{(t)}$.
\begin{lemma} \label{lem:inftynorm-T^11}
	The quantum channel $T^{(1,1)}$ is such that
	\[ \underset{\rho\in D(d^2),\, \rho\perp\psi}{\sup} \left\|T^{(1,1)}(\rho)\right\|_{\infty} = \frac{1}{d^2-1}. \]
\end{lemma}

\begin{proof}
By equation \eqref{eq:T^11-action}, we see that, for any state $\rho$ orthogonal to $\ketbra{\psi}{\psi}$, $T^{(1,1)}(\rho) = Q/(d^2-1)$, so that $\|T^{(1,1)}(\rho)\|_{\infty}=1/(d^2-1)$.
\end{proof}

We then need the technical result below, which is the analogue of Lemma \ref{lem:Bernoulli-T^t} in the study of $T^{(t)}$, and which is as well an immediate corollary of \cite[Lemma 5]{Aubrun2009}, recalled earlier as Lemma \ref{lem:aubrun}.
\begin{lemma} \label{lem:Bernoulli-T^11}
	Let $U_1,\ldots,U_n\in U(d)$. For $\varepsilon_1,\ldots,\varepsilon_n$ independent Bernoulli random variables, we have
	\begin{align*} & \E \left( \underset{\rho\in D(d^2),\, \rho\perp\psi}{\sup} \left\| \sum_{i=1}^n \varepsilon_i U_i\otimes\bar{U}_i \rho U_i^*\otimes\bar{U}_i^* \right\|_{\infty} \right)  \\
	& \leq C(\log d)^{5/2}(\log n)^{1/2} \underset{\rho\in D(d^2),\, \rho\perp\psi}{\sup} \left\| \sum_{i=1}^n U_i\otimes\bar{U}_i \rho U_i^*\otimes\bar{U}_i^* \right\|_{\infty}^{1/2}, 
	\end{align*}
	where $C>0$ is a universal constant.
\end{lemma}

\begin{proof} 
	This follows directly from \cite[Lemma 5]{Aubrun2009}, applied with $d^2-1$ playing the role of $d$ and $U_i\otimes\bar{U}_i$ playing the role of $U_i$, $1\leq i\leq n$.
\end{proof} 

With Lemmas \ref{lem:inftynorm-T^11} and \ref{lem:Bernoulli-T^11} at hand it is straightforward to prove Proposition \ref{prop:11design-T^11-2}, starting from the same symmetrization trick than the one which allows to prove Theorem \ref{th:tdesign-T^t} from Lemmas \ref{lem:inftynorm-T^t} and \ref{lem:Bernoulli-T^t}. We therefore do not repeat the proof here.

So we can now combine Propositions \ref{prop:11design-T^11-1} and \ref{prop:11design-T^11-2} to get Theorem \ref{th:11design-T^11}. It is interesting to note that Propositions \ref{prop:11design-T^11-1} and \ref{prop:11design-T^11-2} give us approximation results for the channel $T^{(1,1)}$ in $(1\to \infty)$-norm, on the maximally entangled state and on states which are orthogonal to it. However, when combining them in order to deal with the case of input states supported on both subspaces, we are only able to get an approximation result in $(1\to 1)$-norm. Indeed, as it will be clear in the proof, in order to show that the approximation error is small also for mixed terms, we need to use the approximation result from Theorem \ref{th:tdesign-T^t} for the channel $T^{(1)}$. Now, since the latter acts on $L(d)$, and not $L(d^2)$ as $T^{(1,1)}$, the approximation error that we can guarantee for it is not small enough to give an interesting approximation result for $T^{(1,1)}$ in the strong $(1\to \infty)$-norm, which is why we have to relax to the weaker $(1\to 1)$-norm. A way around this limitation would probably be to try and prove an analogue of \cite[Lemma 5]{Aubrun2009} which encompasses the action of the channel $T^{(1,1)}$ on the whole input space, rather than analysing separately its action on two subspaces, as we do here.

\begin{proof}[Proof of Theorem \ref{th:11design-T^11}]
By convexity of $\|\cdot\|_1$ and extremality of pure states amongst all states, it is enough to prove that the result is true for all pure input states. Given $\ket{\varphi}$ a unit vector, we can write it as $\ket{\varphi}=\alpha\ket{\psi}+\beta\ket{\psi'}$, where $\alpha=\braket{\psi}{\varphi}$, $|\alpha|^2+|\beta|^2=1$ and $\ket{\psi'}$ is a unit vector orthogonal to $\ket{\psi}$. Defining 
\begin{equation} \label{eq:Delta}
\Delta: X\in L(d^2) \mapsto \frac{1}{n}\sum_{i=1}^n U_i\otimes\bar{U}_i X U_i^*\otimes\bar{U}_i^* - T^{(1,1)}(X) \in L(d^2),  
\end{equation}
we then have
\begin{align*}
\|\Delta(\ketbra{\varphi}{\varphi})\|_1 & = \| |\alpha|^2\Delta(\ketbra{\psi}{\psi}) + |\beta|^2\Delta(\ketbra{\psi'}{\psi'}) + \alpha\bar{\beta}\Delta(\ketbra{\psi}{\psi'}) + \bar{\alpha}\beta\Delta(\ketbra{\psi'}{\psi}) \|_1 \\
& \leq |\alpha|^2 \| \Delta(\ketbra{\psi}{\psi}) \|_1 + |\beta|^2 \| \Delta(\ketbra{\psi'}{\psi'}) \|_1 + 2|\alpha||\beta| \| \Delta(\ketbra{\psi}{\psi'}) \|_1 .
\end{align*}

First, we know from Proposition \ref{prop:11design-T^11-1} that $\| \Delta(\ketbra{\psi}{\psi}) \|_1=0$, while we know from Proposition \ref{prop:11design-T^11-2} that, with probability at least $3/4$, $\| \Delta(\ketbra{\psi'}{\psi'}) \|_1 \leq d^2\| \Delta(\ketbra{\psi'}{\psi'}) \|_{\infty} \leq \epsilon$ for any $\ket{\psi'}$ orthogonal to $\ket{\psi}$. Second, we know that we can write $\ket{\psi'}=X\otimes\mathbf{1}\ket{\psi}$ for some $X$ such that $\Tr (X)=0$ and $\|X\|_2=\sqrt{d}$. That way, since for any $U\in U(d)$, $U\otimes\bar{U}\ket{\psi}=\ket{\psi}$ and $UX\otimes\bar{U}\ket{\psi}=UXU^*\otimes\mathbf{1}\ket{\psi}$, we get
\begin{align*}  
\|\Delta(\ketbra{\psi}{\psi'})\|_1 & = \left\| \ketbra{\psi}{\psi} \, \left( \frac{1}{n}\sum_{i=1}^n U_i X U_i^* - T^{(1)}(X) \right)\otimes\mathbf{1} \right\|_1 \\
& \leq \left\| \left( \frac{1}{n}\sum_{i=1}^n U_i X U_i^* - T^{(1)}(X) \right) \otimes\mathbf{1} \right\|_{\infty} \\
& = \left\| \frac{1}{n}\sum_{i=1}^n U_i X U_i^* - T^{(1)}(X) \right\|_{\infty}.
\end{align*}
Now, we know from Theorem \ref{th:tdesign-T^t} (for $t=1$) that, with probability at least $3/4$, for any $X$ such that $\|X\|_2=\sqrt{d}$,
\[ \left\| \frac{1}{n}\sum_{i=1}^n U_i X U_i^* - T^{(1)}(X) \right\|_{\infty} \leq \frac{\epsilon}{d}\|X\|_1 \leq \frac{\epsilon}{\sqrt{d}}\|X\|_2 = \epsilon  \]
(actually as soon as $n\geq Cd(\log d)^6/\epsilon^2$, hence a fortiori for $n\geq Cd^2(\log d)^6/\epsilon^2$).

Putting everything together we eventually obtain that, with probability at least $1/2$, for any $\ket{\varphi}$,
\[ \|\Delta(\ketbra{\varphi}{\varphi})\|_1 \leq \frac{3\epsilon}{2}, \]
which, up to re-labelling $3\epsilon/2$ in $\epsilon$, is exactly what we wanted to prove.
\end{proof}

\begin{remark}
	It can be shown that the result of Theorem \ref{th:11design-T^11} is optimal, up to a $\mathrm{poly}(\log d)$ factor, just as the one of Theorem \ref{th:tdesign-T^t}. Indeed, using the same reasoning as in Remark \ref{rem:optimality}, together with Lemma \ref{lem:inftynorm-T^11}, we see that, if a channel $\hat{T}^{(1,1)}$ is $\epsilon$-close to $T^{(1,1)}$ in $(1\to 1)$-norm, then it has to satisfy $r(\hat{T}^{(1,1)})\geq (1-\epsilon)(d^2-1)$.
\end{remark}

\subsection{Approximating the twirling super-channel $\Theta$} 

We are now interested in a slightly different kind of twirling, namely one that acts on channels rather than states. We thus define the quantum super-channel $\Theta$ on $\C^d$ as
\begin{equation}
\Theta: \mathcal{M}\in\mathcal{L}(d) \mapsto \left( \Theta(\mathcal{M}):X\in L(d) \mapsto \int_{U\in U(d)} U\mathcal{M}(U^*XU)U^*dU \right) .
\end{equation}

Similarly as before, we here want to show that $\Theta$ can be approximated by sampling `few' unitaries from a `simple' probability measure. We will be able to prove approximation in completely bounded one-to-one norm (also known as diamond norm) for all input channel.

More precisely, denoting by $\mathrm{id}:L(d)\to L(d)$ the identity map on $L(d)$, we will show the following result:
\begin{theorem} \label{th:11design-Theta}
	Let $0<\epsilon<1$. Assume that the probability measure $\mu$ on $U(d)$ is a $2$-design, and let $U_1,\ldots,U_n$ be sampled independently from $\mu$. There exists a universal constant $C>0$ such that, if $n\geq Cd^2(\log d)^6/\epsilon^2$, then with probability at least $1/2$, we have, for all $\mathcal{N}\in\mathcal{C}(d)$,
	\[ \forall\ \rho\in D(d^2),\ \left\| \frac{1}{n}\sum_{i=1}^n  \mathbf{1}\otimes U_i \, \mathrm{id}\otimes\mathcal{N}(\mathbf{1}\otimes U_i^* \, \rho \, \mathbf{1}\otimes U_i) \, \mathbf{1}\otimes U_i^* - \mathrm{id}\otimes\Theta(\mathcal{N})(\rho) \right\|_{1} \leq \epsilon. \]
\end{theorem}

\begin{proof}
By convexity of $\|\cdot\|_1$ and extremality of pure states amongst all states, it is enough to prove that the result is true for all pure input states (and all input channels). Let $\mathcal{N}$ be a channel and $\ket{\varphi}$ be a pure state, which we can write as $\ket{\varphi}= X\otimes\mathbf{1}\ket{\psi}$ for some $X$ such that $\|X\|_2=\sqrt{d}$. Now, for any $U\in U(d)$, $X\otimes U^* \ket{\psi}= X\bar{U}\otimes\mathbf{1}\ket{\psi}$, so that
\begin{align*} 
\mathbf{1}\otimes U \mathrm{id}\otimes\mathcal{N}(\mathbf{1}\otimes U^*  \ketbra{\varphi}{\varphi} \mathbf{1}\otimes U) \mathbf{1}\otimes U^* & = \mathbf{1}\otimes U \mathrm{id}\otimes\mathcal{N}(X\otimes U^* \ketbra{\psi}{\psi} X^*\otimes U ) \mathbf{1}\otimes U^* \\
& = \mathbf{1}\otimes U  \mathrm{id}\otimes\mathcal{N}(X\bar{U}\otimes\mathbf{1} \ketbra{\psi}{\psi} \bar{U}^*X^*\otimes\mathbf{1} ) \mathbf{1}\otimes U^* \\
& = X\bar{U}\otimes U \mathrm{id}\otimes\mathcal{N}(\ketbra{\psi}{\psi}) \bar{U}^*X^*\otimes U^* .
\end{align*}
Therefore, defining $\Delta$ as in equation \eqref{eq:Delta}, we have
\begin{align} \label{eq:approx-Theta}
& \left\| \frac{1}{n}\sum_{i=1}^n  \mathbf{1}\otimes U_i \, \mathrm{id}\otimes\mathcal{N}(\mathbf{1}\otimes U_i^* \, \rho \, \mathbf{1}\otimes U_i) \, \mathbf{1}\otimes U_i^* - \mathrm{id}\otimes\Theta(\mathcal{N})(\rho) \right\|_{1} \nonumber \\
& = \| X\otimes\mathbf{1} \, \Delta(\mathrm{id}\otimes\mathcal{N}(\ketbra{\psi}{\psi})) \, X^*\otimes\mathbf{1} \|_1. 
\end{align}
We now proceed exactly as in the proof of Theorem \ref{th:11design-T^11}. First, by Proposition \ref{prop:11design-T^11-1}, $\Delta(\ketbra{\psi}{\psi})=0$, so that 
\[ \| X\otimes\mathbf{1} \, \Delta(\ketbra{\psi}{\psi}) \, X^*\otimes\mathbf{1} \|_1 = 0 . \]
Second, by Proposition \ref{prop:11design-T^11-2}, with probability at least $3/4$, for any $\ket{\psi'}$ orthogonal to $\ket{\psi}$, $\| \Delta(\ketbra{\psi'}{\psi'}) \|_{\infty}\leq \epsilon/d^2$, so that
\begin{align*} 
\| X\otimes\mathbf{1} \, \Delta(\ketbra{\psi'}{\psi'}) \, X^*\otimes\mathbf{1} \|_1 & \leq \|X\otimes\mathbf{1}\|_2 \|X^*\otimes\mathbf{1}\|_2 \|\Delta(\ketbra{\psi'}{\psi'})\|_{\infty} \\
& = \|X\|_2^2 \|\mathbf{1}\|_2^2\|\Delta(\ketbra{\psi'}{\psi'})\|_{\infty} \\
& \leq \epsilon,  
\end{align*}
where the first inequality is by H\"{o}lder inequality while the last inequality is simply recalling that $\|X\|_2=\|\mathbf{1}\|_2=\sqrt{d}$. Third, any $\ket{\psi'}$ orthogonal to $\ket{\psi}$ can be written as $\ket{\psi'}=Y\otimes\mathbf{1}\ket{\psi}$ for some $Y$ such that $\Tr (Y)=0$ and $\|Y\|_2=\sqrt{d}$. Since for any $U\in U(d)$, $U\otimes\bar{U}\ket{\psi}=\ket{\psi}$ and $UY\otimes\bar{U}\ket{\psi}=UYU^*\otimes\mathbf{1}\ket{\psi}$, we then get
\begin{align*}
\| X\otimes\mathbf{1} \, \Delta(\ketbra{\psi}{\psi'}) \, X^*\otimes\mathbf{1} \|_1 & = \left\| X\otimes\mathbf{1} \, \ketbra{\psi}{\psi} \, \left( \frac{1}{n}\sum_{i=1}^n U_i Y^* U_i^* - T^{(1)}(Y^*) \right) \otimes\mathbf{1} \, X^*\otimes\mathbf{1} \right\|_1 \\
& \leq \left\| X\otimes\mathbf{1} \, \ket{\psi} \right\| \left\| X\otimes\mathbf{1} \left( \frac{1}{n}\sum_{i=1}^n U_i Y U_i^* - T^{(1)}(Y) \right) \otimes\mathbf{1} \, \ket{\psi} \right\| \\
& = \left\| X \left( \frac{1}{n}\sum_{i=1}^n U_i Y U_i^* - T^{(1)}(Y) \right) \otimes\mathbf{1} \, \ket{\psi} \right\| \\
& \leq \left\| X \left( \frac{1}{n}\sum_{i=1}^n U_i Y U_i^* - T^{(1)}(Y) \right) \otimes\mathbf{1} \right\|_{\infty} \\
& = \left\| X \left( \frac{1}{n}\sum_{i=1}^n U_i Y U_i^* - T^{(1)}(Y) \right)  \right\|_{\infty} \\
& \leq \left\| X \right\|_{\infty} \left\| \frac{1}{n}\sum_{i=1}^n U_i Y U_i^* - T^{(1)}(Y) \right\|_{\infty} \\
\end{align*}
where the second equality is because $\| X\otimes\mathbf{1} \ket{\psi} \|=\|\ket{\varphi}\|=1$. Now on the one hand $\|X\|_{\infty}\leq \|X\|_2=\sqrt{d}$. And on the other hand, by Theorem \ref{th:tdesign-T^t} for $t=1$ and $\epsilon/\sqrt{d}$ instead of $\epsilon$, we get that, for $n\geq Cd^2(\log d)^6/\epsilon^2$, with probability at least $3/4$, for all $Y$ such that $\|Y\|_2=\sqrt{d}$,
\[ \left\| \frac{1}{n}\sum_{i=1}^n U_i Y U_i^* - T^{(1)}(Y) \right\|_{\infty} \leq \frac{\epsilon}{d\sqrt{d}}\|Y\|_1 \leq \frac{\epsilon}{d}\|Y\|_2 =\frac{\epsilon}{\sqrt{d}} . \]
And thus, with probability at least $3/4$, for any $\ket{\psi'}$ orthogonal to $\ket{\psi}$,
\[ \| X\otimes\mathbf{1} \, \Delta(\ketbra{\psi}{\psi'}) \, X^*\otimes\mathbf{1} \|_1 \leq \frac{3\epsilon}{2} . \]
Putting everything together, we obtain that, with probability at least $1/2$, for any state $\sigma$ (in particular for $\sigma=\mathrm{id}\otimes\mathcal{N}(\ketbra{\psi}{\psi})$),
\[ \| X\otimes\mathbf{1} \, \Delta(\sigma) \, X^*\otimes\mathbf{1} \|_1 \leq \frac{3\epsilon}{2} . \]
Inserting this into equation \eqref{eq:approx-Theta}, and re-labelling $3\epsilon/2$ in $\epsilon$, yields exactly the claimed result.
\end{proof}	

\subsection{An alternative formulation} 

Given a linear map acting on a normed vector space (that of either operators or super-operators in our case), it is natural to define its so-called \textit{induced norms}. For a linear map $\mathcal{M}:L(d)\to L(d)$, the most relevant induced norm is the one-to-one norm, as well as its completely bounded counterpart (also known as diamond norm). These are defined as
\[ \|\mathcal{M}\|_{1\to 1} = \sup_{X\in L(d),\ \|X\|_1\leq 1} \|\mathcal{M}(X)\|_1 \ \text{and}\ \|\mathcal{M}\|_{\diamond}= \sup_{k\in\N} \|\mathrm{id}_k\otimes\mathcal{M}\|_{1\to 1}, \]
where $\mathrm{id}_k:L(k)\to L(k)$ denotes the identity map on $L(k)$. By extension, for a linear map $\Xi: \mathcal{L}(d)\to\mathcal{L}(d)$, the most relevant induced norm is the diamond-to-diamond norm, as well as its completely bounded counterpart (which we denote with a double diamond). These are defined as 
\[ \|\Xi\|_{\diamond\to\diamond}=\sup_{\mathcal M\in\mathcal L(d),\ \|\mathcal M\|_\diamond\le 1} \left\|\Xi(\mathcal M\right)\|_\diamond\ \text{and}\ \|\Xi\|_{\diamond\diamond}=\sup_{k\in\N}	\|\mathrm{id}_k \otimes \Xi \|_{\diamond\to\diamond}, \]
where $\mathrm{id}_k:\mathcal L(k)\to\mathcal L(k)$ denotes the identity map on $\mathcal L(k)$. 

Using these definitions, we can reformulate Theorems \ref{th:tdesign-T^t} and \ref{th:11design-Theta} as follows:
\begin{corollary} \label{cor:tdesign-T^t}
Let $0<\epsilon<1$. Assume that the probability measure $\mu$ on $U(d)$ is a $t$-design, and let $U_1,\ldots,U_n$ be sampled independently from $\mu$. There exists a universal constant $C>0$ such that, if $n\geq C(td)^t(t\log d)^6/\epsilon^2$, then with probability at least $1/2$, we have
\[ \left\| T^{(t)} - T^{(t)}_{\mu,n} \right\|_{1\to 1} \leq \epsilon, \]
where
\[ T^{(t)}_{\mu,n}:X\in L(d^t) \mapsto \frac{1}{n}\sum_{i=1}^n U_i^{\otimes t}XU_i^{* \otimes t}.  \]
\end{corollary}
\begin{corollary} \label{cor:11design-Theta}
	Let $0<\epsilon<1$. Assume that the probability measure $\mu$ on $U(d)$ is a $2$-design, and let $U_1,\ldots,U_n$ be sampled independently from $\mu$. There exists a universal constant $C>0$ such that, if $n\geq Cd^2(\log d)^6/\epsilon^2$, then with probability at least $1/2$, we have
	\[ \left\|\Theta-\Theta_{\mu,n}\right\|_{\diamond\to\diamond}\le  \epsilon, \]
	where 
	\[ \Theta_{\mu,n}:  \mathcal{M}\in\mathcal{L}(d) \mapsto \left( \Theta_{\mu,n}(\mathcal{M}): X\in L(d) \mapsto \frac{1}{n}\sum_{i=1}^n U_i\mathcal{M}(U_i^*XU_i)U_i^* \right). \]
\end{corollary}

For the applications in the next section, it is natural to define a \emph{$k$-bounded} variant of the completely bounded one-to-one and diamond-to-diamond norms, i.e. 
	\[ \|\mathcal{M}\|_{\diamond,k}= \|\mathrm{id}_k\otimes\mathcal{M}\|_{1\to 1} \ \text{and}\  \|\Xi\|_{\diamond\diamond,k}=\|\mathrm{id}_k\otimes\Xi\|_{\diamond\to\diamond}. \]
	By the Schmidt decomposition it is clear that $\|\cdot\|_{\diamond,d}=\|\cdot\|_{\diamond}$ and $\|\cdot\|_{\diamond\diamond,d^2}=\|\cdot\|_{\diamond\diamond}$. What is more, it is well-known that, for any $k\leq d$, $\|\cdot\|_{\diamond,k}\leq k\|\cdot\|_{1\to 1}$. We now prove a similar upper bound for $\|\cdot\|_{\diamond\diamond,k}$ in terms of $\|\cdot\|_{\diamond \to\diamond}$. 
\begin{lemma} \label{lem:k-bounded}
	For any linear map $\Xi: \mathcal{L}(d)\to\mathcal{L}(d)$, we have
	\[ \|\Xi\|_{\diamond\diamond,k}\le k^2\|\Xi\|_{\diamond\to\diamond} . \]
\end{lemma}

\begin{proof}
	Let $\mathcal M\in \mathcal L(kd)$ with $\|\mathcal M\|_\diamond =1$ be such that
	\[ \|\Xi\|_{\diamond\diamond,k}= \| \mathrm{id}_k\otimes\Xi\|_{\diamond\to\diamond} = \|\left(\mathrm{id}_k\otimes\Xi\right)(\mathcal M)\|_{\diamond}. \]
	By concavity of the diamond norm, we can assume that $\mathcal M$ has a Choi matrix $\eta_{\mathcal{M}}\in L(k^2d^2)$ of rank one, i.e.~$\eta_{\mathcal M}=\ketbra{\varphi_{\mathcal M}}{\phi_{\mathcal M}}$ for some $\ket{\varphi_{\mathcal M}},\ket{\phi_{\mathcal M}}\in\mathbf{C}^{k^2d^2}$. Let us now write $\ket{\varphi_{\mathcal M}},\ket{\phi_{\mathcal M}}\in\mathbf{C}^{k^2}\otimes\C^{d^2}$ in their Schmidt decomposition: 
	\[ \ket{\varphi_{\mathcal M}}=\sum_{i=1}^{k^2}\sqrt{p_i}\ket{\alpha_i\beta_i} \ \text{and}\ \ket{\phi_{\mathcal M}}=\sum_{i=1}^{k^2}\sqrt{q_i}\ket{\gamma_i\zeta_i},\]
	with $\{p_i\}_{1\leq i\leq k^2}$, $\{q_i\}_{1\leq i\leq k^2}$ subnormalized probability distributions and with $\{\ket{\alpha_i}\}_{1\leq i\leq k^2}$, $\{\ket{\gamma_i}\}_{1\leq i\leq k^2}$ and $\{\ket{\beta_i}\}_{1\leq i\leq k^2}$, $\{\ket{\zeta_i}\}_{1\leq i\leq k^2}$ orthonormal families in $\C^{k^2}$ and $\C^{d^2}$ respectively. We can hence write
	\[\mathcal M=\sum_{i,j=1}^{k^2}\sqrt{p_ip_j} \mathcal{M}_{ij}^k\otimes \mathcal{M}_{ij}^d, \]
	where $\mathcal{M}_{ij}^k:\mathcal{L}(k)\to\mathcal{L}(k)$ is defined as having Choi matrix $\eta_{\mathcal{M}_{ij}^k}=\ketbra{\alpha_i}{\gamma_j}\in L(k^2)$ and $\mathcal{M}_{ij}^d:\mathcal{L}(d)\to\mathcal{L}(d)$ is defined as having Choi matrix $\eta_{\mathcal{M}_{ij}^d}=\ketbra{\beta_i}{\zeta_j}\in L(d^2)$.
	We then have by the triangle inequality
	\[ \|\Xi\|_{\diamond\diamond,k} = \|\left(\mathrm{id}_k\otimes\Xi\right)(\mathcal M)\|_{\diamond} \leq \sum_{i,j=1}^{k^2} \sqrt{p_iq_j} \| \mathcal M_{ij}^k \otimes \Xi(\mathcal M_{ij}^d) \|_{\diamond} = \sum_{i,j=1}^{k^2} \sqrt{p_iq_j} \|\Xi(\mathcal M_{ij}^d) \|_{\diamond}. \]
	To finish the proof, we then simply have to observe that
	\[  \sum_{i,j=1}^{k^2} \sqrt{p_iq_j}\|\Xi(\mathcal M_{ij}^d) \|_{\diamond} \leq \|\Xi\|_{\diamond\to\diamond} \sum_{i,j=1}^{k^2}\sqrt{p_iq_j} = \|\Xi\|_{\diamond\to\diamond} \left(\sum_{i=1}^{k^2}\sqrt{p_i}\right) \left(\sum_{j=1}^{k^2}\sqrt{q_j}\right) \leq k^2 \|\Xi\|_{\diamond\to\diamond}, \]
	where the first inequality is by definition of the diamond-to-diamond norm and the second inequality is due to the Cauchy-Schwarz inequality. 
\end{proof}

\section{Application: Quantum non-malleable encryption against adversaries with small quantum memory}

Information-theoretically secure quantum encryption has been studied extensively. In particular, the one-time variants of security goals such as confidentiality, authenticity and non-malleability have been defined for quantum encryption. When assessing the efficiency of a symmetric-key encryption scheme, there are three main figures of merit, the running time of the encryption and decryption algorithms, the ciphertext length and the key length. Here, we focus on the latter two figures of merit. Protocols have been designed which achieve the optimal scaling with respect to key length (up to log factors). More precisely, the results are as follows. The quantum one-time pad scheme, that encrypts a quantum system by applying a random element of the Pauli group, requires $2\log d$ bits of key \cite{Ambainis2000}. The quantum authentication scheme presented in \cite{Barnum2002} uses $2\log d+O(s)$ bits of key and $\log d+O(s)$ bits of ciphertext to achieve $s$ bits of security. And non-malleable encryption with unitaries (hence with plaintext space and ciphertext space being the same) can be done with $(4+o(1))\log d$ bits of key \cite{Ambainis2009}. Here we describe a construction for non-malleable encryption \emph{without adversarial side information} with unitaries using $2\log{d}+O(\log\log d)$ bits of key. In addition, our scheme has confidentiality against adversaries \emph{with side information}. In other words, it is an alternative to the standard quantum one-time pad with the added property of non-malleability without side information at only an additive logarithmic cost in terms of key length. 

\subsection{One-time-secure quantum encryption} 

We begin by defining more rigorously the different cryptographic notions mentioned above. In the following, given a a finite set $\mathcal X$, the notation $\mathbf E_{x\in\mathcal X}$ is used to denote the expectation value of a random variable $x$ distributed uniformly on $\mathcal X$.

\begin{definition}[Quantum encryption scheme]
	A triple $(\mathcal X, \Enc, \Dec)$ where
	\begin{enumerate}
		\item[i)]  $\mathcal X$ is a finite set,
		\item[ii)] $\{\Enc_x\}_{x\in\mathcal X}$ is a family of quantum channels $\Enc_x: L(d_M)\to L(d_C)$,
		\item[iii)]  $\{\Dec_x\}_{x\in\mathcal X}$ is a family of quantum channels $\Dec_x: L(d_C)\to L(d_M)$ 
	\end{enumerate}
is called quantum encryption scheme if
 \[\forall\ x\in\mathcal{X},\ \Dec_x\circ\Enc_x=\mathrm{id}.\]
\end{definition}
The parameters $\log_2|\mathcal X|$, $\log_2d_M$ and $\log_2d_C$ are called key length, message length and ciphertext length, respectively.

Given a quantum state $\sigma\in D(d)$, define the quantum channel $\langle\sigma\rangle\in\mathcal C(d)$ by $\langle\sigma\rangle(X)=\Tr(X)\sigma$.

\begin{definition}[Indistinguishability of ciphertexts]
	A quantum encryption scheme has $\epsilon$-indistinguishable ciphertexts, if there exists a quantum state $\sigma\in D(d_C)$ such that
	\begin{equation*}
		\left\| \mathbf E_{x\in\mathcal X}[\Enc_x]-\langle\sigma\rangle \right\|_\diamond\le \epsilon.
	\end{equation*}
	A quantum encryption scheme has $\epsilon$-indistinguishable ciphertexts against adversaries without side information if the above inequality holds with the diamond norm replaced by the one-to-one norm.
\end{definition}

\begin{definition}[Non-malleability]
	A quantum encryption scheme is $\epsilon$-non-malleable, if there exists a quantum state $\sigma\in D(d_C)$ such that for all side information dimension $d_E$ and all $\Lambda\in \mathcal C(d_Cd_E)$ there exist completely positive maps $\Lambda_=,\Lambda_{\neq}\in\mathcal L(d_E)$ whose sum is trace-preserving and $p\in[0,1]$ such that
	\begin{equation*}
	\left\| \mathbf E_{x\in\mathcal X}\left[\left(\Dec_x\otimes \mathrm{id}\right)\circ\Lambda\circ\left(\Enc_x\otimes \mathrm{id}\right)\right]-\left(p\,\mathrm{id}\otimes \Lambda_{=}+(1-p)\langle\sigma\rangle\otimes \Lambda_{\neq}\right) \right\|_\diamond\le \epsilon.
	\end{equation*}
	A quantum encryption scheme is $\epsilon$-non-malleable against adversaries without side information, if there exists a quantum state $\sigma\in D(d_C)$ such that for all $\Lambda\in \mathcal C(d_C)$ there exists $p\in[0,1]$ such that
	\begin{equation*}
	\left\| \mathbf E_{x\in\mathcal X}\left[\Dec_x\circ\Lambda\circ\Enc_x\right]-\left( p\,\mathrm{id}+(1-p)\langle\sigma\rangle \right) \right\|_\diamond\le \epsilon.
	\end{equation*}
\end{definition}

\subsection{Non-stabilized norms and adversaries without quantum side information} 

Any family of unitary matrices $\{U_x\}_{x\in\mathcal X}$ defines a quantum encryption scheme via, for all $x\in\mathcal X$, $\Enc_x(X)=U_xXU_x^*$ and $\Dec_x=\Enc_x^*$.  For such unitary quantum encryption schemes, it is easy to see that $\epsilon$-indistinguishability of ciphertexts implies that the family of unitaries is a $2\epsilon$-approximate $1$-design in diamond norm, and any $\epsilon$-approximate $1$-design in diamond norm gives rise to a quantum encryption scheme with $\epsilon$-indistinguishable ciphertexts. The weaker property of $\epsilon$-indistinguishability of ciphertexts against adversaries without side information and the $\epsilon$-approximate $1$-design property measured in one-to-one norm have the same relationship.

Similarly, if a unitary quantum encryption scheme is $\epsilon$-non-malleable, then it is a $2\epsilon$-approximate channel twirl in completely bounded diamond-to-diamond norm, and an $\epsilon$-approximate channel twirl in completely bounded diamond-to-diamond norm gives rise to a quantum encryption scheme that is $\epsilon$-non-malleable \cite{Ambainis2009,Alagic2017}. Again, the weaker $\epsilon$-non-malleability against adversaries without side information and the $\epsilon$-approximate channel twirl property measured in diamond-to-diamond norm have the same relationship.

The results in the previous section thus immediately imply the following for random unitary encryption schemes:
\begin{theorem} \label{th:no-SI}
	Let $0<\epsilon<1$. Assume that the probability measure $\mu$ on $U(d)$ is a $2$-design, and let $U_1,\ldots,U_n$ be sampled independently from $\mu$. There exists a universal constant $C>0$ such that, if $n\geq Cd^2(\log d)^6/\epsilon^2$, then with probability at least $1/2$, the quantum encryption scheme defined by the family of unitaries $\{U_1,\ldots,U_n\}$ has $\epsilon/\sqrt{d}$-indistinguishable ciphertexts and is $\epsilon$-non-malleable against adversaries without side information.
\end{theorem}
\begin{proof}
	Let us define 
	\begin{equation} \label{eq:T-Theta-proof}
	T_{\mu,n}^{(1)}: X\in L(d) \mapsto \frac{1}{n}\sum_{i=1}^n U_iXU_i^* \ \text{and}\ \Theta_{\mu,n}: \mathcal{M}\in\mathcal{L}(d) \mapsto \frac{1}{n}\sum_{i=1}^n U_i\mathcal{M}(U_i^*\cdot U_i)U_i^* . 
	\end{equation}
	To begin with notice that, if $\mu$ is a $2$-design then it is a fortiori a $1$-design. Hence by Corollary \ref{cor:tdesign-T^t} (for $t=1$ and $\epsilon/\sqrt{d}$ instead of $\epsilon$) and Corollary \ref{cor:11design-Theta}, the probability that $T_{\mu,n}^{(1)}$ is an $\epsilon/\sqrt{d}$-approximate $1$-design in one-to-one norm and the probability that $\Theta_{\mu,n}$ is an $\epsilon$-approximate channel twirl in diamond-to-diamond norm are both at least $3/4$ for $n\geq Cd^2(\log d)^6/\epsilon^2$. By the union bound, both properties hold simultaneously with probability at least $1/2$. And as explained before, if this is so then the corresponding unitary quantum encryption scheme has $\epsilon/\sqrt{d}$-indistinguishable ciphertexts and is $\epsilon$-non-malleable against adversaries without side information. 
\end{proof}

Using the result of Lemma \ref{lem:k-bounded}, relating the $k$-bounded diamond norm to the one-to-one norm and the $k$-bounded double diamond norm to the diamond-to-diamond norm, we can immediately derive from Theorem \ref{th:no-SI} a generalisation of it that applies to the case where the adversary has side information, but in bounded quantity. 

\begin{corollary}
	Let $0<\epsilon<1$. Assume that the probability measure $\mu$ on $U(d)$ is a $2$-design, and let $U_1,\ldots,U_n$ be sampled independently from $\mu$. There exists a universal constant $C>0$ such that, if $n\geq Cd^2(\log d)^6k^4/\epsilon^2$, then with probability at least $1/2$, the quantum encryption scheme defined by the family of unitaries $\{U_1,\ldots,U_n\}$ has $\epsilon/k\sqrt{d}$-indistinguishable ciphertexts and is $\epsilon$-non-malleable against adversaries with $k$-bounded side information.
\end{corollary}
\begin{proof}
	Let $T_{\mu,n}^{(1)},\Theta_{\mu}$ be defined as in equation \eqref{eq:T-Theta-proof}. We have shown in the proof of Theorem \ref{th:no-SI} that, for $n\geq Cd^2(\log d)^6/\epsilon^2$, with probability larger than $1/2$,
	\[ \left\| T^{(1)}-T_{\mu,n}^{(1)} \right\|_{1\to 1}\leq \frac{\epsilon}{\sqrt{d}} \ \text{and}\ \left\| \Theta-\Theta_{\mu,n} \right\|_{\diamond \to \diamond}\leq \epsilon . \]
	Now by Lemma \ref{lem:k-bounded}, we know that this implies that 
	\[ \left\| T^{(1)}-T_{\mu,n}^{(1)} \right\|_{\diamond,k}\leq \frac{\epsilon k}{\sqrt{d}} \ \text{and}\ \left\| \Theta-\Theta_{\mu,n} \right\|_{\diamond\diamond,k}\leq \epsilon k^2 . \]
	Hence redefining $\epsilon$ as $\epsilon k^2$, we get that, for $n\geq Cd^2(\log d)^6k^4/\epsilon^2$, with probability larger than $1/2$, $T_{\mu,n}^{(1)}$ is an $\epsilon/k\sqrt{d}$-approximate $1$-design in $k$-bounded diamond norm and $\Theta_{\mu,n}$ is an $\epsilon$-approximate channel twirl in $k$-bounded double diamond norm. And if this is so then the corresponding unitary quantum encryption scheme has $\epsilon/k\sqrt{d}$-indistinguishable ciphertexts and is $\epsilon$-non-malleable against adversaries with $k$-bounded side information. 
\end{proof}

\subsection{A note on efficiency}

While our scheme is more efficient in terms of key length and in terms of encryption and decryption given the element of the design that needs to be applied (if instantiated with an efficiently implementable $2$-design, such as e.g.~the Clifford group), specifying the randomly chosen subset of the exact $2$-design is inefficient. This is a problem shared by all schemes based on the sub-sampling technique, i.e.~in particular by the ones constructed in \cite{Hayden2004} and \cite{Ambainis2009}. To construct \emph{efficiently specifiable} approximate designs in the weak norms we consider, that are still smaller than approximate designs in the diamond norm, additional new techniques seem to be necessary. A possible approach would for instance be to analyse random quantum circuits with respect to these norms. Indeed, all results showing that random quantum circuits of a given size are expected to be approximate $t$-designs, following the seminal work \cite{Brandao2016}, use a metrics which is stronger than the one we need for our cryptographic applications. It is thus probable that, for the latter, shorter random quantum circuits are already working well.

It is also worth pointing out that our results can be easily generalized to the case where the unitaries are sampled from an approximate rather than exact design. For instance, in the case of Theorem \ref{th:tdesign-T^t} we would have the following result: If $\mu$ is an $\epsilon/d^t$-approximate $t$-design in $(1\to\infty)$-norm, then we can obtain a $2\epsilon/d^t$-approximate $t$-design in $(1\to\infty)$-norm by sampling $C(td)^t(t\log d)^6/\epsilon^2$ unitaries from $\mu$. Indeed, the proof of Theorem \ref{th:tdesign-T^t} relates the behaviour of the sampled twirling channel $T^{(t)}_{\mu,n}$ to that of its average $T^{(t)}_{\mu}$, independently of whether or not this average is the same as if taken over the Haar measure, i.e.~equal or not to $T^{(t)}$. Once you have proven that $T^{(t)}_{\mu,n}$ is close to $T^{(t)}_{\mu}$, you just have to use that, by assumption on $\mu$, $T^{(t)}_{\mu}$ is close to $T^{(t)}$, and add the two approximation errors. This provides a strategy to circumvent the difficulty of constructing exact $t$-designs for $t>3$, since on the contrary efficient constructions of approximate ones (even in a stronger sense than the one we require) are known.

\section*{Acknowledgements}

C.M.~thanks David Gross for discussions on $t$-designs. We also thank the two anonymous referees of this work for their numerous and insightful comments, which truly helped in improving it. C.L.~acknowledges financial support from the French CNRS (project PEPS JCJC). C.M. was supported by  a NWO VIDI grant (Project No.~639.022.519) and a NWO VENI grant (Project No.~VI.Veni.192.159).

\addcontentsline{toc}{section}{References}

\printbibliography

\end{document}